\documentclass[twocolumn,10pt,journal]{IEEEtran}
\usepackage{indentfirst,amsmath}
\usepackage{color,rotating,pdflscape,verbatim}
\usepackage{cite}
\usepackage{amsfonts,amsmath,amssymb,amsthm}
\usepackage{latexsym,amscd}
\usepackage{amsbsy,extarrows}
\usepackage{graphicx,multirow,bm}
\usepackage[table,dvipsnames]{xcolor}
\usepackage{extarrows}
\usepackage{tikz}
\usetikzlibrary{arrows,shapes}
\usetikzlibrary{shadows}

\newtheorem{Theorem}{Theorem}
\newtheorem{Corollary}{Corollary}
\newtheorem{Definition}{Definition}
\newtheorem{Example}{Example}

\newtheorem{Remark}{Remark}
\newtheorem{Lemma}{Lemma}

\newtheorem{Construction}{Construction}
\newtheorem{Proposition}{Proposition}

\newcommand{\F}{\ensuremath{\mathbb{F}}}

\usepackage{mathrsfs}
\newcommand{\Gen}{\mathsf{Gen}}
\newcommand{\Enc}{\mathsf{Enc}}
\newcommand{\Dec}{\mathsf{Dec}}

\newenvironment{psmallmatrix}
  {\left(\begin{smallmatrix}}
  {\end{smallmatrix}\right)}

\newcommand\ch[1]{\textcolor{magenta}{Cami: #1}}
\newcommand\jl[1]{\textcolor{cyan}{Jie: #1}}

\usepackage[innermargin=0.545in,outermargin=0.545in,top=0.6in,bottom=.6in]{geometry}
\begin{document}
\title{Efficient Recovery of a Shared Secret via Cooperation: Applications to SDMM and PIR
}

\author{Jie Li,~\IEEEmembership{Member,~IEEE}, Okko Makkonen,
        Camilla Hollanti,~\IEEEmembership{Member,~IEEE},~and Oliver W. Gnilke
\thanks{
Manuscript received June 1, 2021; revised November 1, 2021; accepted December 21, 2021. The work of J. Li was supported in part by the National Science Foundation of China under Grant No. 61801176. The work of C. Hollanti and O. Makkonen was supported by the Academy of Finland, under Grants No. 318937 and 336005.}
\thanks{J. Li was with the Department of Mathematics and Systems Analysis,
                    Aalto University, FI-00076 Aalto,  Finland, and also with the Hubei Key Laboratory of Applied Mathematics, Faculty of Mathematics and Statistics, Hubei University, Wuhan 430062, China (e-mail: jieli873@gmail.com).}
\thanks{O. Makkonen and C. Hollanti are with the Department of Mathematics and Systems Analysis,
                    Aalto University, FI-00076 Aalto,  Finland  (e-mails: \{okko.makkonen,camilla.hollanti\}@aalto.fi).}
                    \thanks{O. Gnilke is with the Department of Mathematical Sciences, University of Aalborg, 9220 Aalborg \O, Denmark (email: owg@math.aau.dk).}
}
\date{}
\maketitle

\begin{abstract}
This work considers the problem of privately outsourcing the computation of a matrix product  over a finite field $\F_q$ to $N$ helper servers. These servers are considered to be honest but curious, \emph{i.e.}, they behave according to the protocol but will try to deduce information about the user's data. Furthermore, any set of up to $X$ servers is allowed to share their data. Previous works considered this collusion a hindrance and the download cost of the schemes increases with growing $X$. We propose to utilize such linkage between servers to the user's advantage by allowing servers to cooperate in the computational task. This leads to a significant gain in the download cost for the proposed schemes. The gain naturally comes at the cost of increased communication load between the servers. Hence, the proposed cooperative schemes can be understood as outsourcing both computational cost and communication cost. Both information--theoretically secure and computationally secure schemes are considered, showing that allowing information leakage that is computationally hard to utilize will lead to further gains.

The proposed server cooperation is then exemplified for specific secure distributed matrix multiplication (SDMM) schemes and linear private information retrieval (PIR). Similar ideas naturally apply to many other use cases as well, but not necessarily always with lowered costs.
\end{abstract}


\section{Introduction}
Matrix multiplication is one of the key operations in many science and engineering fields, such as machine learning and cloud computing. Carrying  out the computation on powerful distributed servers is desirable for improving efficiency, as the user can partition the computational task into several sub-tasks and outsource them to many servers. By scaling out computations across many distributed servers, security concerns arise. This raises the problem of secure distributed matrix multiplication (SDMM), which has recently received a lot of attention from an information--theoretic perspective, both in terms of code constructions and capacity bounds  \cite{dutta2018unified,chang2018capacity,dutta2020optimal,yang2019secure,kakar2019capacity,jia2019capacity,d2021degree,d2020gasp,yu2020entangled,aliasgari2020private,jia2021cross,zhu2021improved}. 
Private information retrieval (PIR) --- another problem currently getting a lot of attention --- can be seen as a special case of SDMM in the typical case of the responses being linear functions \cite{chor1995private,sun2017capacity,sun2017capacityrobust,banawan2018capacity,tajeddine2018private,freij2017private,freij2018t,kumar2019achieving,zhu2019new,zhou2020capacity,jia2020x,wang2019symmetric,d2019one,li2020towards,holzbaur2022towards,song2021equivalence,zhu2021capacity}.
Indeed, a connection between a variant of SDMM and a form of PIR was drawn in \cite{jia2019capacity}, where an upper bound of the capacity of a special form of SDMM was characterized by that of PIR. Later in \cite{d2020notes}, it was shown that the problem of SDMM  can  be  efficiently  solved  via the PIR  method if the download cost is the only performance metric. However, the computational cost to the servers and the upload cost are very high. It was further shown in \cite{d2020notes} that polynomial codes can indeed improve the total time complexity as long as the parameters are carefully chosen.

Generally, three performance metrics are of particular interest when designing a distributed computation scheme (\emph{e.g.}, matrix multiplication).
\begin{itemize}
  \item The upload cost:  the amount of data transmitted from the user to the servers to assign the sub-tasks;
  \item The download cost: the amount of data to be downloaded from the servers;
  \item Recovery threshold $R_c$:  the minimal number of servers that need to complete their tasks before the user can recover the desired computation result.
\end{itemize}
In the case of secure distributed matrix multiplication,  the upload cost can be essentially determined once the matrix partition method is chosen and the number of servers is given.

\subsection{Related Work}
Existing SDMM codes are mostly based on various matrix partitioning and coding techniques used in \cite{yu2017polynomial,dutta2018unified,dutta2020optimal,yu2020straggler},  and abundant tradeoffs among the aforementioned performance metrics have been derived. In \cite{yu2020straggler,dutta2020optimal,dutta2018unified}, the matrices  $A\in \mathbb{F}_q^{t\times s}$ and $B\in \mathbb{F}_q^{s\times r}$ are partitioned as
\begin{equation}\label{Eqn_partition_AB}
 A=\left(A_{k,j}\right)_{k\in [0,m), j\in [0, p)},
   \,\, B=\left(B_{j',k'}\right)_{j'\in [0,p), k'\in [0, n)},    
\end{equation}
where $m |t$, $p | s$, and $n|r$, which subsumes the so-called \textit{inner  product partitioning (IPP)} (when $m=n=1$ in \eqref{Eqn_partition_AB}) and the \textit{outer  product partitioning (OPP)} (when $p=1$ in \eqref{Eqn_partition_AB}) as special cases \cite{d2021degree}. In \cite{jia2019capacity,jia2021cross,chen2021gcsa}, the so-called cross-subspace-alignment codes are employed to outsource the subtasks.

In \cite{chang2018capacity}, Chang and Tandon proposed constructions of SDMM codes and  addressed  the capacity when only one of the matrices $A$, $B$ is required to be $X$-secure. Yang and Lee \cite{yang2019secure} proposed SDMM codes in the case of $X=1$. Kakar \emph{et al.} \cite{kakar2019capacity} and D'Oliveira \emph{et al.} \cite{d2020gasp,d2021degree} further provided more efficient constructions over that of Chang and Tandon.  Based on the matrix partitioning technique in \cite{yu2020straggler,dutta2020optimal}, Aliasgari \emph{et al.} \cite{aliasgari2020private}   presented more general SDMM codes. Yu \emph{et al.} \cite{yu2020entangled} recently proposed an SDMM scheme based on the Entangled Polynomial code in \cite{yu2020straggler} and the bilinear complexity \cite{blaser2013fast} for multiplying two matrices. Mital \emph{et al.} \cite{mital2020secure} proposed an SDMM scheme with discrete Fourier transform, where IPP is employed and the download cost is shown to be smaller than some previous ones. However, it cannot mitigate stragglers. Based on the general matrix partitioning in \eqref{Eqn_partition_AB}, another scheme which  
provides a slightly weaker ‘group-wise’ straggler-robustness was also provided in \cite{mital2020secure}.  Table \ref{comp} gives a summary of the aforementioned results.

\begin{table*}[htbp]
\begin{center}
\caption{A summary of the key parameters of some known SDMM schemes.}
\label{comp}\setlength{\tabcolsep}{3pt}
\begin{tabular}{|c|c|c|c|c|}
\hline  & Upload cost & Download cost & $R_c$ (Recovery Threshold)& References \\
\hline Chang-Tandon code  & $N(\frac{ts}{m}+\frac{sr}{m})$ & $R_c\frac{tr}{m^2}$ & $(m+X)^2$ & \cite{chang2018capacity}\\
\hline Kakar \emph{et al.} code 1  & $N(\frac{ts}{m}+\frac{sr}{n})$ & $R_c\frac{tr}{mn}$ & $(m+X)(n+1)-1$ & \cite{kakar2019capacity}\\
\hline GASP code  & $N(\frac{ts}{m}+\frac{sr}{n})$ & $R_c\frac{tr}{mn}$ & $\ge mn+\max\{m, n\}+2X-1$ & \cite{d2021degree}\\
\hline SGPD code  & $N(\frac{ts}{mp}+\frac{sr}{pn})$ & $R_c\frac{tr}{mn}$ & $\left\{
  \begin{array}{ll}
 pmn+pm+pn\lceil\frac{X}{p}\rceil+2X-1,  & p<m,\\
pmn+pm+(mn-m)\lceil\frac{X}{\min\{m,n\}}\rceil+2X-1, & p\ge m.\\
  \end{array}
\right. $ & \cite{aliasgari2020private} \\
\hline Secure Entangled Polynomial code  & $N(\frac{ts}{mp}+\frac{sr}{pn})$ & $R_c\frac{tr}{mn}$ & $2R(p,m,n)+2X-1$ & \cite{yu2020straggler} \\
\hline  Mital \emph{et al.} code 1 &  $N(\frac{ts}{p}+\frac{sr}{p})$ & $R_ctr$ &$p+2X (=N)$ & \cite{mital2020secure} \\
\hline
\end{tabular}
\end{center}
\end{table*}

As a connection to X-secure T-PIR (XTPIR for short), which was first studied in \cite{jia2020x}, we can assume that the matrix $B$ is stored across the servers, while the user has a private matrix $A$. In this case, one can regard the matrix $B$ as a collection of $s$ messages, each of length $r$, \emph{i.e.}, a row of $B$ denotes a message. If $A$ is some row of the $s\times s$ identity matrix, then computing $AB$ is equivalent to retrieving one message from the database. Therefore,  any such SDMM scheme yields an XTPIR scheme with $T=X$ (where $X$ servers can collude to deduce the message identity and the message data stored in the servers) by treating the encoded shares of $A$ as the queries sent from the user to the servers, and the encoded shares of $B$ as data stored across the servers \cite{jia2019capacity}. For a general matrix $A$,
computing $AB$ is equivalent to retrieving $t$ linear combinations of the $s$ messages, which can be regarded as a generalization of XTPIR and a kind of private computation \cite{sun2018capacity}.

There are some other SDMM models besides the above ones, in \cite{jia2019capacity}, the problem of SDMM was generalized to the case that there are two matrix batches $A^{(0)}, A^{(1)},\ldots,A^{(L-1)}$ and $B^{(0)}, B^{(1)},\ldots,B^{(L-1)}$, where it is assumed that the matrices are mainly stored across the servers while the user may or may not have some side information related to the matrices. The user wishes to obtain the products $A^{(i)}B^{(i)}$ for all $i\in [0, L)$ with the assistance of $N$ distributed servers. This problem variant is termed  \textit{secure distributed batch matrix multiplication (SDBMM)}.

Although there are a lot of studies on SDMM, its fundamental limits remain an open problem. Up to now, only a few results are known. In \cite{chang2018capacity},  \cite{jia2019capacity}, and some other works, they only focus on minimizing the download cost, and aim to increase the SDMM rate, which is defined as the ratio of the size of the desired information to the download cost. The SDMM capacity is then the supremum of SDMM rate over all feasible schemes. In \cite{chang2018capacity}, the capacity of one-sided SDMM (\textit{i.e.,} only  one  of  the  matrices $A$, $B$ is required to be $X$-secure) is $\frac{R_c-X}{R_c}$. In \cite{jia2019capacity}, converse bounds of the SDBMM capacity under their model (\textit{i.e.}, the matrices are mainly stored across the servers while the user may or may not have some side information related to the matrices) were obtained, and were shown to be tight in some special cases. To the best of our knowledge, no other results on the upper bound of the SDMM capacity have been reported in the literature except the above ones.

\subsection{Motivation and Contributions}

An interesting observation  is that in the case of any $X$ servers colluding to deduce information about either $A$ or $B$ by sharing their data, previous work considered this collusion a hindrance and the download cost of the schemes increase 
with growing $X$.  It seems that, to date, no one has considered that such server connectivity could be utilized to the user's  advantage. For example, in the presence of connectivity or bandwidth
constraints, the server--user and server--server communication
costs differ depending on their mutual proximity. This has been assumed and widely studied in erasure coding for clustered architectures and its variations \cite{joonas,gaston2013realistic,pernas2013non,prakash2018storage,hou2019rack,chen2020explicit}, where it is assumed that  the servers are organized into several clusters, and communication within the clusters is free or much cheaper than that between different clusters. In this paper, we adopt this idea to the SDMM problem, and assume that the servers and the user have such an architecture that the communication between the helper servers is very cheap so that it can be neglected. This assumption is of course not always valid, in which case the increased inter-server communication should be accounted for. Outsourcing the computation to servers also helps in avoiding network congestion between the servers and the user.

In this paper, we propose to use a possible link between servers to the user's advantage by allowing servers to cooperate in the computational task. This leads to a significant reduction in the  download cost for the proposed scheme, and helps to avoid network congestion between the servers and the user.
The reduced download cost does not come for free, but the price is paid by the servers, \emph{i.e.}, is achieved by offloading the communication between servers and the user to the communication between servers.

In more detail, the contributions of this paper can be  summarized as:
\begin{itemize}
    \item To the best of our knowledge, it is the first time to consider the case that the (download)  communication cost of the user is partially outsourced to the servers, similarly as is done for the computation. This is enabled by allowing the servers to cooperate in addition to colluding.
This is particularly attractive in applications where the server--user and server--server communication costs differ. This could be the case when the helper servers are clustered within close proximity. 

\item We present several examples of how communication cost can be outsourced to helper nodes. We present these examples in the language of secret sharing, since the data retrieval phase in SDMM and PIR can be seen as data collection and recovery in a secret sharing scheme \cite{song2021equivalence}.

\item We study an  SDMM scheme from an information--theoretic perspective that lends itself  particularly well to the use of cooperation among servers. We show that this cooperative scheme outperforms several other non-cooperative schemes for a wide range of parameters. An XTPIR scheme under the server cooperation model is also introduced, and can achieve a larger PIR rate than the non-cooperative scheme.

\item We also consider an encryption-based cooperative SDMM scheme, which achieves a better rate at the cost of losing information--theoretic security. We show that such a scheme is still computationally secure against colluding servers.

\end{itemize}

The idea of utilizing cooperation between colluding servers to reduce communication cost can be applied to other schemes as well. However, it is highly non-trivial how to organize the servers and their communication to optimize this cooperation.

We note that another form of server cooperative  SDMM has also been studied but in the context of secure multi-party computation in \cite{akbari2021secure,chen2021gcsa} recently. This model contains three parties: the source nodes, the server nodes (helpers), and a master node (the user). The data matrices are outsourced  by the source nodes, and the aim is to prevent the master node from learning anything about the inputs besides the result of the computation. This is achieved by employing secure multi-party computation, where server cooperation is used as a key ingredient  but it does not help in reducing the download cost.   The problem settings  in \cite{akbari2021secure,chen2021gcsa} and the one in this paper are quite different, thus it is not meaningful to compare these schemes even though server cooperation is employed in some form in all of them. 

\subsection{Organization}
The paper is organized as follows.
In Section \ref{sec:pre}, we introduce some  preliminaries, including the so-called collusion and cooperation graph for the abstraction of the problem, the computation problem setting, Reed-Solomon codes, secret sharing, and symmetric encryption. In Section \ref{sec:colla}, we introduce information--theoretically secure cooperative data retrieval, and present examples of how cooperation between servers can significantly reduce the download cost for the user. Furthermore, a motivating example of the SDMM code under the server cooperation model and the general scheme are presented, followed by an extensive comparison with existing schemes, an XTPIR scheme under the server cooperation model is also introduced.  In Section \ref{sec:enc_coop}, we introduce computationally secure cooperative schemes based on encryption and show how existing SDMM schemes from an information--theoretical perspective can be converted to computationally secure SDMM schemes under the server cooperation model. Finally, Section \ref{sec:conclusion} draws the conclusions.

\section{Preliminaries and Problem Setting}
\label{sec:pre}
In this section, we introduce some necessary preliminaries. First of all, let us fix some notation. Let $\mathbb{F}_q$ denote a finite field containing $q$ elements, where $q$ is a prime power. For two integers $a$ and $b$ with $a\le b$, denote by $[a, b)$ the set $\{a,a+1,\ldots,b-1\}$.

\subsection{Collusion and Cooperation}\label{Sec:collusion_graph}
Let $V$ be the set of helper servers. Then we define the collusion (resp. cooperation) graph $(V,E)$ on the vertex set $V$ by adding an edge $(v_i,v_j)$ to $E$ whenever servers $v_i$ and $v_j$ collude (resp. cooperate). We only consider the symmetric case here where both servers gain access to each other's data, but similar results can be formulated for the directed case.

\begin{Definition}
A scheme is secure against $X$-collusion if no information is leaked to the servers even if there is $X$-collusion, \emph{i.e.}, any $X$ servers may exchange their received messages. This is equivalent to the biggest collusion component of $(V,E)$ being of size $X$.
\end{Definition}

\begin{Remark} As is customary (also cf. $t$-PIR schemes where any $t$-set may collude), it is assumed that collusion is non-transitive. In other words, if $X_1$ is a colluding set and $X_2$ is another colluding set with $X_1\cap X_2\neq\emptyset$, it does \emph{not} follow that $X_1\cup X_2$ would also be a colluding set.
\end{Remark}

\subsection{Problem Setting}
Assume that the user wishes to recover a shared secret, for example,  the product of two matrices (\textit{i.e., the SDMM problem}) or one out of many files that are stored across several servers (\textit{i.e., the PIR problem}). The user needs assistance from the servers to complete the task, but does not want to reveal any related information to the colluding servers. 

In the following, we consider these two specific problems for motivation. First, we restrict the problem setting to SDMM under the server cooperation model. Then, one can immediately find that the  PIR problem setting with cooperation is the same as SDMM once we introduce it in matrix form in Section \ref{sec:IT-XTPIR}. The setting for a general cooperative recovery of a shared secret is then similar. 
In SDMM under server cooperation model, the user is interested in computing the product of  two matrices $A\in \mathbb{F}_q^{t\times s}$ and $B\in \mathbb{F}_q^{s\times r}$  over some  finite field $\mathbb{F}_q$ with the  assistance of $N$ honest-but-curious servers. The contents of the matrices $A$ and $B$ should remain secret in the information--theoretic sense, \emph{i.e.}, the
servers learn nothing about the content of the matrices even if $X$ of them  collude, where $1\le X<N$.
The multiplication of the two matrices $A$ and $B$ can be accomplished according to the following steps.
\begin{itemize}
\item Upload phase: The user encodes $A$ and $B$ with some random matrices to obtain matrices $\tilde{A}_i$ and $\tilde{B}_i$ for $i\in [0, N)$,  and then sends $\tilde{A}_i$ and $\tilde{B}_i$ to server $i$. In addition, the user sends  the respective evaluation points $a_i\in\mathbb{F}_q, i\in [0, N)$  to the servers, which will be used to generate the shared data in the next step\footnote{As these evaluation points are always small compared to the matrices, the cost of uploading them can be neglected. It may also be possible to assume that this information is a priori shared with the helper servers. For this reason and also for the convenience of notation, the upload cost of the additional data $a_i, i\in [0, N)$ will be neglected in the sequel.}.
\item Computation and cooperative phase:  First, server $i$ computes the product of $\tilde{A}_i$ and $\tilde{B}_i$ for all $i\in [0, N)$. Second, the fastest $R_c$ servers seek to cooperate by forming cooperation groups. Each server in a cooperation group sends some data to a representative server of the cooperation group. Let $\mathcal{M}_{j\rightarrow j'}$ denote the data that server $j$ sends to server $j'$, and $\mathcal{M}_j$ denote all the data that server $j$ receives from the other servers.
\item Decoding phase: A representative server in each cooperation group $i$ sends a message to the user, denoted by $Y_i$ for $i\in \mathcal{R}\subset [0, N)$, where $|\mathcal{R}|$ is the number of cooperation groups. The user then decodes $AB$ from what they received.
\end{itemize}

The scheme must satisfy the following two constraints.
\begin{itemize}
  \item $X$-security: Any $X$ servers learn nothing about neither $A$ nor $B$, \textit{i.e.,}
\begin{eqnarray*}
  I(\{\tilde{A}_i, \tilde{B}_i\}_{i\in \mathcal{X}}; A, B)=0
\end{eqnarray*}
for any $\mathcal{X}\subset [0, N)$ with $|\mathcal{X}|=X$, where $I(Y;Z)$ denotes the mutual information between $Y$ and $Z$.
  \item Correctness: The user must be able to recover $AB$ from $Y_i$, $i\in \mathcal{R}$, \textit{i.e.,}
\begin{eqnarray*}
  H(AB|Y_i, i\in \mathcal{R})=0,
\end{eqnarray*}
where $H(Y|Z)$ denotes the entropy of $Y$ conditioned on $Z$.
\end{itemize}

In addition, we wish to minimize the recovery threshold $R_c$ as well as the communication cost, which is comprised of  the upload cost, the download cost, and cooperation cost (\emph{i.e.}, communication cost among servers). As usual in the SDMM setting we define the upload cost as $$\sum\limits_{i=0}^{N-1}\left(|\tilde{A}_i|+|\tilde{B}_i|\right),$$ where $|M|$ is the size of the matrix $M$ counted as $\F_q$ symbols.  
Similarly, the download cost is defined as   $\sum\limits_{i\in \mathcal{R}}|Y_i|$. The cooperation cost is defined as  $\sum\limits_{i\in \mathcal{R}}|\mathcal{M}_i|$.
\begin{Remark}
For PIR it is common to measure these quantities using entropy instead of size, assuming a compression can be applied before transfer. For the SDMM problem though that we want to focus on, requiring the user to decompress the received information can be equivalent or harder than the computational task that was outsourced.
\end{Remark}

The problem setting for XTPIR via cooperation is similar, as we see in Section  \ref{sec:IT-XTPIR}.

\subsection{Reed-Solomon (RS) Codes}
The encoding phase of SDMM codes is usually based on Reed-Solomon codes or their sub-codes. To this end, we briefly introduce RS codes in the following.

\begin{Definition}(\cite{reed1960polynomial})
Let $x_0,x_1,\ldots,x_{N-1}$ be $N$ distinct elements in $\mathbb{F}_q$. The $[N,K]$ RS code associate to these $N$ locators is defined as
\begin{equation*}
\mathcal{C}=\{\left(p(x_0),p(x_1),\ldots,p(x_{N-1})\right) \mid p\in \mathbb{F}_q[x], \deg(p)<K\}.
\end{equation*}
In other words, a generator matrix of the $[N,K]$ RS code can be given as
\begin{equation*}
    G=\begin{pmatrix}
1 & 1 & \cdots & 1\\
x_0 & x_1 & \cdots & x_{N-1}\\
\vdots & \vdots & \ddots & \vdots\\
x_0^{K-1} & x_1^{K-1} & \cdots & x_{N-1}^{K-1}
    \end{pmatrix}.
\end{equation*}
\end{Definition}
The above generator matrix is the transpose of a Vandermonde matrix over $\mathbb{F}_q$, showing that RS codes belong to the class of MDS codes.

\subsection{Secret Sharing}
Secret sharing has been introduced independently in \cite{shamir1979share} and \cite{blakley1979safeguarding}. A good introduction to the general theory of secret sharing can be found in \cite{LecNotesSecret}.
We begin by describing a very general setup that can be found in many applications.

\begin{Proposition}[Secure Coded Storage]\label{Prop_PrivStor}
Let $C$ be a linear $[N,K]$ code with generator matrix $G$. Let $G_{>{\rho}}$ be the matrix consisting of the lowest $K-{\rho}$ rows of $G$. The matrix $G_{>{\rho}}$ defines an $[N,K-{\rho}]$ code and we denote the minimum distance of its dual by $X+1$. Then we can securely store ${\rho}$ symbols on a storage system consisting of $N$ servers using the code $C$ such that any $X$ servers learn nothing about the ${\rho}$ information symbols.

\begin{proof}
We begin by describing the storage and then prove its secrecy. Let $m=(m_0,\dots, m_{\rho-1})$ be the vector containing our information, and $s=(s_0, \dots, s_{K-{\rho}-1})$ be a random vector in $\F_q^{K-{\rho}}$. We encode $(m, s)$ into a length $N$ vector by calculating
\[ (m_0,\dots, m_{\rho-1},s_0,\dots,s_{K-{\rho}-1})G=(y_0,\dots,y_{N-1}) \]
and store the symbol $y_i$ on server $i$.

Now consider any set of $X$  servers that might collude with index set $T$. Together they observe the partial vector $y_{|T}=(y_t)_{t \in T}$ which is given by
\[ y_{|T}=m {(G_{\leq {\rho}})}_{|T} + s {(G_{>{\rho}})}_{|T}. \]
Since the dual of $G_{>{\rho}}$ has minimum distance $X+1$, any set of $X$ columns of $G_{>{\rho}}$ is linearly independent. Hence the vector $s {(G_{>{\rho}})}_{|T}$ is uniformly random in $\F_q^{X}$ and the servers can not learn anything about the vector $m$.
\end{proof}
\end{Proposition}

\begin{Corollary}\label{CorXsecure}
Let $C$ be an $[N,K]$ RS code. Then we can securely store $\rho < K$ symbols on a database consisting of $N$ servers using the code $C$, such that any $X=K-\rho$ servers learn nothing about the information symbols.
\end{Corollary}

Many secret sharing schemes can be understood as secure coded storage in the sense of Proposition \ref{Prop_PrivStor}. Furthermore, some PIR schemes and SDMM schemes have as an intermediate step a secure coded storage where the servers hold coded shares of a secret message. The following considerations hence apply to all of these schemes.

\begin{Definition}[Linear Secret Sharing]\label{Def_linear_secret}
 A secret sharing scheme is considered to be linear if the secret $m$ is a linear combination of any $K$ shares.
\end{Definition}

\begin{Example}[Shamir Secret Sharing]\label{ex_Shamir}
We describe the Shamir Secret Sharing scheme and explain how it fits our definition of secure coded storage. Let $x_0, \dots, x_{N-1}$ be $N$ distinct elements of $\F_q$. We see that
\[ G= \begin{pmatrix} x_0^{K-1} & \cdots & x_{N-1}^{K-1}\\
    \vdots & \ddots& \vdots\\
    x_0& \cdots & x_{N-1}\\
    1 & \cdots & 1\\
\end{pmatrix}
\] is a generator matrix for an $[N,K]$ RS code. Clearly, $G_{>1}$ defines an $[N,K-1]$ RS code and its dual hence has minimum distance $K$. Given a secret $m$ and a random vector $u\in \F_q^{K-1}$ we calculate $N$ shares
    $(m,u)G=(y_0, \dots, y_{N-1})$.
This defines an $[N,K]$ threshold secret sharing scheme, where any $K-1$ servers can not learn anything about the secret $m$ but any $K$ can successfully recover the secret.
\end{Example}

\begin{Remark}
The unusual way of defining the generator matrix is not necessary, but this way the restriction that $0$ cannot be used as an evaluation point $x_i$ is removed and the proof that the lower part of the matrix defines an MDS code becomes trivial.
\end{Remark}

\subsection{Symmetric Encryption}
In order to allow cooperative SDMM schemes with more flexible matrix partitioning, we will also consider computational security instead of information--theoretic security. We will allow some information leakage that will be computationally hard to harness. This is enabled by utilizing a symmetric encryption scheme to aid secret sharing as described below. For more information, please refer to \cite{katz2020introduction, goldreich2005foundations}.

A symmetric encryption algorithm consists of three probabilistic polynomial-time algorithms $(\Gen, \Enc, \Dec)$ along with the key space $\mathcal{K}$, message space $\mathcal{M}$, and the ciphertext space $\mathcal{C}$. The spaces are parametrized with the security parameter $n$, \emph{e.g.}, the key space $\mathcal{K}$ consists of spaces $\mathcal{K}_n$ for $n \geq 1$. Often the key space, message space and ciphertext space are spaces of binary strings parametrized by the length of the string. Here we are interested in encryption algorithms over strings of finite field elements, \emph{i.e.}, $\mathcal{K}, \mathcal{M}, \mathcal{C} = \{ \F_q^n \}_{n \geq 1}$. 

The algorithm $\Gen$ takes the security parameter $1^n$ and outputs a key $k \in \mathcal{K}_n$ according to some probability distribution. Then, the algorithm $\Enc$ takes the key $k$ and a message $m \in \mathcal{M}$, and outputs an encryption $c = \Enc_k(m) \in \mathcal{C}$. Finally, the algorithm $\Dec$ takes the key $k$ and a ciphertext $c$, and outputs a plaintext value $\Dec_k(c) \in \mathcal{M}$ such that $\Dec_k(\Enc_k(m)) = m$. These algorithms run in polynomial time of the length of their input.

The strongest security requirement for an encryption scheme is the perfect information--theoretic security that requires that the mutual information between the ciphertext and the message is zero. The one-time pad achieves this notion of security. In this case any adversary is allowed to be arbitrarily powerful without being able to gain any information. Another notion of security is \emph{semantic security}, which states that no \emph{efficient} adversary is able to compute any additional information about the message given the ciphertext with a \emph{non-negligible} probability. Here efficient means probabilistic polynomial-time in the length of the input, and a negligible function is a function that is asymptotically smaller than any positive polynomial. The definition of semantic security is weaker than information--theoretic security, but it is a good analogue when information--theoretic security is unobtainable. In \cite{goldwasser1984probabilistic} it was shown that semantic security is equivalent to having indistinguishable encryptions under chosen plaintext attack (IND-CPA). 

\begin{Definition}\label{Def_ind-cpa}
A symmetric encryption scheme is IND-CPA secure when no efficient adversary can distinguish the encryptions of two chosen messages, when given oracle access to a decryption function. 
\end{Definition}

An important building block for designing encryption algorithms is a pseudorandom function (PRF).

\begin{Definition}\label{Def_PRF}
A pseudorandom function is a collection of functions $\{f_n\}$ parametrized by the security parameter, where $f_n \colon \F_q^n \times \F_q^n \to \F_q^n$ is deterministic, efficiently computable, and no efficient adversary can distinguish the partial function $f_n(k, \cdot)$ from a truly random function $\F_q^n \to \F_q^n$ with non-negligible probability. Here the probability is taken over the key space $\F_q^n$ and the randomness of the truly random function.
\end{Definition}

Additionally, we can define a variable output-length pseudorandom function, which produces an output of a given length. Such a PRF can be constructed from a fixed output-length PRF using counter mode, which runs the PRF using multiple inputs, which can produce arbitrarily many random values.
The following construction defines a symmetric encryption scheme from a variable length PRF.

\begin{Construction}\label{Con_enc}
This symmetric encryption scheme uses a variable output-length pseudorandom function $f$.
\begin{itemize}
    \item $\Gen(1^n)$ outputs a uniformly random key $k \in \F_q^n$.
    \item Given the key $k$ and message $m$ of length $|m|$, $\Enc$ draws $r \in \F_q^n$ uniformly at random and computes $z = f_n(k, r, |m|)$, \textit{i.e.} $z \in \F_q^{|m|}$ is pseudorandom. Then $\Enc$ outputs $(r, m + z)$.
    \item Given the key $k$ and a ciphertext $(r, c)$, $\Dec$ outputs $c - f_n(k, r, |c|)$.
\end{itemize}
\end{Construction}

Construction \ref{Con_enc} defines an IND-CPA secure encryption scheme. This encryption scheme is often known as a stream cipher, where a pseudorandom value is combined with the plaintext. The one-time pad is a special case of this encryption, where the size of the key is the same size as the message. Using pseudorandom functions, it is possible to reduce the size of the key, which is desirable in many applications. The proof of security can be found in \cite{katz2020introduction}. The IND-CPA security requires that the encryption is randomized, \textit{i.e.}, the same message is encrypted to a different ciphertext each time the algorithm is called. For a weaker security definition it is possible to have deterministic encryption, which means that the random value $r$ is not needed in Construction \ref{Con_enc}.

\section{Information--Theoretically Secure Cooperation}\label{sec:colla}

In this section, we study the efficient recovery of a shared secret from an information--theoretic perspective, with a particular emphasis on SDMM.

Particularly, we assume that the collusion and cooperation graphs coincide, \emph{i.e.}, the servers who cooperate are considered as colluding sets, and vice versa it is expected that colluding sets may also cooperate. The reason for this is that data is sent in an unencrypted form, hence revealing information already when cooperating (with or without interest to collude).  Note however that in both cases $X$ is an upper bound, not necessarily a strict number of colluding/cooperating servers. In addition, we assume that all but one component of the collusion/cooperation graph are of size $X$ and the remaining component contains $|V| \mod X$ servers.

\subsection{Cooperation}
\begin{Lemma}\label{Thm_SSrecover}
Let $y_0, \dots, y_{K-1}$ be a recovery set of a linear secret sharing scheme and $\alpha_i \in \F_q$ be the coefficients of the linear combination resulting in the secret $m$, \emph{i.e.}, $m=\sum \alpha_i y_i$.
Let $V_0, \dots, V_{\gamma-1}$ be the collection of connected components of the collusion graph, \emph{i.e.} $\bigsqcup_{c=0}^{\gamma-1} V_c=V$. Using cooperation between connected servers, it suffices to download $\gamma$ symbols to recover the secret.
\begin{proof}
We describe a recovery scheme that only contacts $\gamma$ servers and hence achieves a download cost of $\gamma$. For any connected component one vertex $v_c$ is selected. The server $v_c$ collects the shares of all servers in its component and calculates a response $r_c:=\sum_{i: v_i \in V_c} \alpha_i y_i$. We see that $\sum_{c=0}^{\gamma-1} r_c=\sum \alpha_i y_i=s$ and hence a data collector recovers the secret by adding the responses $r_c$.
\end{proof}
\end{Lemma}
\begin{Remark}
Note that for the scheme described in Lemma \ref{Thm_SSrecover} it is not necessary to know the collusion graph beforehand. The servers are able to organize the calculation of the responses independently.
\end{Remark}

Obviously, some applications store more than one symbol using secure storage. In these cases, retrieval becomes more challenging. Improvements in the download cost using cooperating groups are still possible, but more care is needed in terms of which servers share their data. This is summarized in the lemma below, while the straightforward proof and examples showcasing the result are delegated to the extended arXiv version \cite{coopSDMM_arxiv}.

\begin{Lemma}
Let $G$ be the generator matrix of the storage code. We partition the set of columns into subsets $G|_{V_c}$ and the servers in the set $V_c$ compute a linear combination $r_c=\sum_{i \in V_c} \alpha_i y_i$. Then a user can recover a secret $m_i$ from these $r_c$ if $e_i$ is in the linear span of the vectors $(\sum_{i \in V_c} \alpha_i g^i)$,  where  $G|_{V_c}$ denotes the sub-matrix of $G$ formed by the columns indicated by $V_c$,  $g^i$ is the $i^{th}$ column of $G$, and $e_i$ is the $i$-th column of the identify matrix with the same order as $G$.
\end{Lemma}

\subsection{Lagrange Interpolation}
Another important class of examples can be described using Lagrange interpolation.

\begin{Definition}\label{def_lagrangeex}  (\cite{stoer2013introduction})
Given a set of $K$ data points
 \begin{equation*}
   (x_{0},y_{0}),\ldots ,(x_{j},y_{j}),\ldots ,(x_{K-1},y_{K-1})
 \end{equation*}
where   $x_{j}$ are pairwise distinct, the \emph{interpolation polynomial in the Lagrange form} is a linear combination
\begin{equation}\label{Eqn_Lag_poly}
  L(x):=\sum\limits_{j=0}^{K-1}y_{j}\ell^{(j)}(x)
\end{equation}
of Lagrange basis polynomials
\begin{equation}\label{Eqn_ell}
  \ell^{(j)}(x):=\prod_{i=0,i\ne j}^{K-1}\frac{x-x_i}{x_j-x_i},
\end{equation}
where $j\in [0, K)$.
\end{Definition}

Rewrite the polynomials $L(x)$ in \eqref{Eqn_Lag_poly} and $\ell^{(j)}(x)$ in \eqref{Eqn_ell} as
\begin{eqnarray}\label{Eqn_poly_L}
L(x)=L_0+L_1x+\ldots+L_{K-1}x^{K-1}
\end{eqnarray}
and
\begin{eqnarray}\label{Eqn_ell2}
\ell^{(j)}(x)=\ell^{(j)}_0+\ell^{(j)}_1x+\ldots+\ell^{(j)}_{K-1}x^{K-1}.
\end{eqnarray}
Then by \eqref{Eqn_Lag_poly}-\eqref{Eqn_ell2}, the coefficient $L_{\theta}$ of $L(x)$  can be expressed as
\begin{equation}\label{Eqn_L_theta}
L_{\theta}=y_0\ell^{(0)}_{\theta}+y_1\ell^{(1)}_{\theta}+\cdots+y_{K-1}\ell^{(K-1)}_{\theta} \mbox{\ for\ }\theta\in [0,K).
\end{equation}
Viewing the coefficient $L_\theta$ as a shared secret, we see that the scheme defined in Def. \ref{def_lagrangeex} is a linear secret sharing scheme in the sense of Def. \ref{Def_linear_secret}.
Thus we have the following result, which can be seen as a realisation of Lemma~\ref{Thm_SSrecover}.

\begin{Proposition}\label{Prop_lag}
Assume there are  $N$ servers, and any $X$ of them   can cooperate. Assume that $y_i$ is the evaluation of some polynomial $L(x)$ at $x_i$ for $i\in [0, N)$, where $\deg(L(x))=K-1$ and $K\in [1, N]$.  If server $i$ has the data  $x_0, \ldots, x_{N-1}$ and $y_i$ for $i\in [0, N)$, then the user can obtain one of the  coefficient $L_{\theta}$ in \eqref{Eqn_poly_L} with the download cost being  $\lceil\frac{K}{X}\rceil |L_{\theta}|$ and cooperation cost being $K|L_{\theta}| - \lceil\frac{K}{X}\rceil |L_{\theta}|$, where $\theta\in [0, K)$.
\end{Proposition}

\begin{proof}
Let $\overline{K}=\lceil \frac{K}{X}\rceil$. W.l.o.g., assume that the first $K$ servers are the fastest ones and the user obtains the desired coefficient from them. For $i\in [0, \overline{K}-1)$, assume that   servers $iX, iX+1,\ldots, iX+X-1$ seek to cooperate with each other, while servers $(\overline{K}-1)X,
\ldots, K-1$ seek to cooperate, then the coefficient $L_{\theta}$   can be retrieved through
 the following two phases.

\begin{itemize}
  \item Computation and cooperative phase: Each server $j$ first computes the polynomial $\ell^{(j)}(x)$ in \eqref{Eqn_ell} and then computes $y_j\ell^{(j)}_{\theta}$.
For $i\in [0, \overline{K}-1)$, server $iX+j$ transmits $y_{iX+j}\ell^{(iX+j)}_{\theta}$ to server $iX$ for $j\in [1, X)$, while server $(\overline{K}-1)X+j$ transmits $y_{(\overline{K}-1)X+j}\ell^{\left((\overline{K}-1)X+j\right)}_{\theta}$ to server $(\overline{K}-1)X$ for $j\in [1, K-(\overline{K}-1)X)$.
For $i\in [0, \overline{K}-1)$, server $iX$ further computes $$Y_{iX}=\sum\limits_{k=iX}^{iX+X-1}y_{k}\ell^{(k)}_{\theta},$$ while server $(\overline{K}-1)X$ further computes $$Y_{(\overline{K}-1)X}=\sum\limits_{k=(\overline{K}-1)X}^{K-1}y_{k}\ell^{(k)}_{\theta}.$$
\item Decoding phase:
For $i\in [0, \overline{K})$, server $iX$ sends $Y_{iX}$ to the user, who then sums the data they received to obtain  $L_{\theta}$ according to \eqref{Eqn_L_theta}.
\end{itemize}

Then the total download cost is $\lceil\frac{K}{X}\rceil |L_{\theta}|$ while the cooperation cost is
\begin{equation*}
\sum\limits_{i=0}^{\overline{K}-2}\sum\limits_{j=1}^{X-1}|L_{\theta}| +\sum\limits_{j=1}^{K-(\overline{K}-1)X-1}|L_{\theta}|=K |L_{\theta}| - \left\lceil\frac{K}{X}\right\rceil |L_{\theta}| 
\end{equation*}
since the amount of data communicated between two servers equals $|L_{\theta}|$ by \eqref{Eqn_Lag_poly}.
\end{proof}

\subsection{A Motivating Example of the SDMM Code under Server Cooperation}
Assume that the user is interested in computing $AB$ for $A\in \mathbb{F}_q^{t\times s}$ and $B\in \mathbb{F}_q^{s\times r}$ with the assistance of $N\ge 7$ servers, while any   $X = 2$ servers can collude to deduce the information of $A$ and $B$, where $t,s,r$ are even.
The user divides the matrices $A$ and $B$ into block matrices as
\begin{equation*}
 A=\begin{pmatrix}    A_{0} & A_{1} \\
   \end{pmatrix}
   \,,\, B=\begin{pmatrix}
       B_{0}  \\
       B_{1} \\
   \end{pmatrix},
\end{equation*}
where $A_{j}\in \mathbb{F}_q^{t\times \frac{s}{2}}$ for $j=0, 1$ and $B_{k}\in \mathbb{F}_q^{\frac{s}{2}\times r}$ for $k=0, 1$. Then
\begin{eqnarray*}
  AB &=&  A_{0}B_0+A_1B_1.
\end{eqnarray*}

Let $Z_{0},  Z_{1}$ be two random matrices over $\mathbb{F}_q^{t\times \frac{s}{2}}$ and $S_{0},  S_{1}$ be two random matrices over $\mathbb{F}_q^{\frac{s}{2}\times r}$. The user encodes the matrices $A$ and $B$ by an $[N, 4]$ RS code, \textit{i.e.,} first creating two polynomials
\begin{eqnarray*}
  f(x) &=& A_{0}+A_1x+Z_0x^{2}+Z_1x^{3}, \\
   g(x) &=& B_{0}x+B_1+S_0x^{2}+S_1x^3,
\end{eqnarray*}
and then evaluating them at $N$ distinct points in $\mathbb{F}_q$, say
$a_0, \ldots, a_{N-1}$. Then the user sends $f(a_i)$,  $g(a_i)$ to server $i$ for $i\in [0, N)$,
\textit{i.e.,} the upload cost is $N(\frac{ts}{2}+\frac{sr}{2})$.

Let $h(x)=f(x)g(x)$, \textit{i.e.,}
{\small
\begin{IEEEeqnarray*}{rCl}
  &&h(x)\\&=&f(x)g(x)\\
    &=&A_{0}B_1+(A_{0}B_{0}+A_1B_1)x+(A_1B_{0}+A_0S_0+Z_0B_1)x^2\\&&+(A_0S_1+A_1S_0+Z_0B_0+Z_1B_1)x^3\\&&+(A_1S_1+Z_1B_0+Z_0S_0)x^4+(Z_0S_1+Z_1S_0)x^5+Z_1S_1x^6.
\end{IEEEeqnarray*}
}Now $\deg(h(x))=6$ and $A_{0}B_{0}+A_1B_1$ is exactly the coefficient of the monomial $x$ in $h(x)$, therefore, the recovery threshold is $R_c=7$.

The coefficient $A_{0}B_{0}+A_1B_1$ of the term $x$ in $h(x)$    can be retrieved through
the following steps.

\begin{itemize}
  \item Computation and cooperative phase:  First, server $j$ computes the product of $f(a_j)$ and $g(a_j)$ to obtain $h(a_j)$ for all $j\in [0, N)$. Assume that servers $j_0, j_1, \ldots, j_6$ are the fastest $7$ servers, and further assume that servers $j_{2i}$ and $j_{2i+1}$ seek to cooperate with each other for $i\in [0, 3)$. Second, each server $j_i$ ($i\in [0, 7)$)  computes the polynomial $$\ell^{(j_i)}(x):=\prod_{u=0,u\ne i}^{6}\frac{x-a_{j_u}}{a_{j_i}-a_{j_u}}$$ to obtain the coefficient $\ell^{(j_i)}_{1}$ of the term $x$ in $\ell^{(j_i)}(x)$, and then multiply it with $h(a_{j_i})$.   Finally,  server $j_{2i+1}$ transmits $\ell^{(j_{2i+1})}_{1}h(a_{j_{2i+1}})$ to server $j_{2i}$ for $i\in [0, \lfloor \frac{7}{2}\rfloor)$, who then further computes
\begin{equation*}
 Y_{j_{2i}}=\ell^{(j_{2i})}_{1}h(a_{j_{2i}})+\ell^{(j_{2i+1})}_{1}h(a_{j_{2i+1}}).
\end{equation*}
\item Decoding phase: Server $j_{2i}$  sends $Y_{j_{2i}}$  to the user for $i\in [0, 4)$, where $Y_{j_6}=\ell^{(j_6)}_{1}h(a_{j_6})$.
The user sums up the data they received  and  obtains
\begin{IEEEeqnarray*}{rCl}
 && A_{0}B_{0}+A_1B_1\\&=&\underbrace{\ell^{(j_0)}_{1}h(a_{j_0})+\ell^{(j_1)}_{1}h(a_{j_1})}_{\rm answer~from~server~ 0}+\underbrace{\ell^{(j_2)}_{1}h(a_{j_2})+\ell^{(j_3)}_{1}h(a_{j_3})}_{\rm answer~from~server~ 2}\\&&+\underbrace{\ell^{(j_4)}_{1}h(a_{j_4})+\ell^{(j_5)}_{1}h(a_{j_5})}_{\rm answer~from~server~4}+\underbrace{\ell^{(j_6)}_{1}h(a_{j_6})}_{\rm answer~from~server~6}
\end{IEEEeqnarray*}
according to \eqref{Eqn_L_theta}.
\end{itemize}
Thus, the download cost is $4tr$ and the cooperation cost is $3tr$.

\subsection{A Cooperative SDMM Scheme}
In this subsection, we propose an SDMM code construction under the server cooperation model based on Matdot codes \cite{dutta2020optimal}, \textit{i.e.,} employing IPP. Let us assume that the user is interested in computing $AB$, where $A\in \mathbb{F}_q^{t\times s}$ and $B\in \mathbb{F}_q^{s\times r}$ with the assistance of $N$ servers, while any    $X$ servers can collude to deduce the information of $A$ and $B$. These colluding servers can now also cooperate.

In general, we have the following result for the new cooperative SDMM code.
\begin{Theorem}\label{Thm_C1}
Assume any $X$   servers can collude and  cooperate. Then, there exists an explicit cooperative SDMM scheme by which the product of   $A$ and $B$ can be securely computed with the assistance of $N$ servers, with upload cost  $N(\frac{ts}{p}+\frac{sr}{p})$, download cost  $tr\lceil\frac{R_c}{X}\rceil$, cooperation cost $tr(R_c-\lceil\frac{R_c}{X}\rceil)$, and   recovery threshold  $R_c=2p+2X-1$.
\end{Theorem}
\begin{proof}
The user partitions the matrices $A$ and $B$ by the inner  product partitioning   as
\begin{equation}\label{Eqn_IPP}
 A=\left(
     \begin{array}{cccc}
       A_{0} & A_{1} & \cdots & A_{p-1} \\
     \end{array}
   \right),\,\, B=\left(
     \begin{array}{c}
       B_{0}  \\
       B_{1}  \\
      \vdots  \\
      B_{p-1} \\
     \end{array}
   \right),
\end{equation}
where $A_{j}\in \mathbb{F}_q^{t\times \frac{s}{p}}$ and $B_{j}\in \mathbb{F}_q^{\frac{s}{p}\times r}$. Then
\begin{eqnarray}\label{Eqn_AB_mn1}
  AB &=&  A_{0}B_{0}+A_{1}B_{1}+\cdots+A_{p-1}B_{p-1}.
\end{eqnarray}

Let $Z_{0}, \ldots, Z_{X-1}$ be $X$ random matrices over $\mathbb{F}_q^{t\times \frac{s}{p}}$ and $S_{0}, \ldots, S_{X-1}$ be $X$ random matrices over $\mathbb{F}_q^{\frac{s}{p}\times r}$.   The user encodes the matrices $A$ and $B$ by an $[N, p+X]$ RS code, \textit{i.e.,} creates two polynomials
\begin{eqnarray*}
  f(x) &=& \sum\limits_{j=0}^{p-1}A_{j}x^{\alpha_{j}}+ \sum\limits_{t=0}^{X-1}Z_tx^{\gamma_t}, \\
   g(x) &=& \sum\limits_{j=0}^{p-1}B_{j}x^{\beta_{j}}+ \sum\limits_{t=0}^{X-1}S_tx^{\delta_t},
\end{eqnarray*}
where  $\alpha_{j}=j$, $\beta_{j}=p-1-j$ for $j\in [0, p)$, $\gamma_t=\delta_t=p+t$ for  $t\in [0,X)$, and then evaluates them at $N$ distinct points
$a_0, \ldots, a_{N-1}$  in $\mathbb{F}_q$. Then, the user sends $f(a_i)$,  $g(a_i)$  to server $i$ for $i\in [0, N)$, 
yielding an upload cost  $N(\frac{ts}{p}+\frac{sr}{p})$.

Let $h(x)=f(x)g(x)$, \textit{i.e.,}
\begin{IEEEeqnarray*}{rCl}
  h(x)&=&\sum\limits_{j=0}^{p-1}\sum\limits_{j'=0}^{p-1}A_{j}B_{j'}x^{\alpha_{j}+\beta_{j'}}+
 \sum\limits_{j=0}^{p-1}\sum\limits_{t'=0}^{X-1}A_{j}S_{t'}x^{\alpha_{j}+\delta_{t'}}\\&&+ \sum\limits_{t=0}^{X-1}\sum\limits_{j'=0}^{p-1}Z_tB_{j'}x^{\gamma_t+\beta_{j'}}+\sum\limits_{t=0}^{X-1}\sum\limits_{t'=0}^{X-1}Z_tS_{t'}x^{\delta_{t}+\delta_{t'}}\\
 &=&\sum\limits_{j=0}^{p-1}A_{j}B_{j}x^{p-1}+\sum\limits_{j=0}^{p-1}\sum\limits_{j'=0,j'\ne j}^{p-1}A_{j}B_{j'}x^{p-1+j-j'}\\&&+
 \sum\limits_{j=0}^{p-1}\sum\limits_{t'=0}^{X-1}A_{j}S_{t'}x^{p+j+t'}+ \sum\limits_{t=0}^{X-1}\sum\limits_{t'=0}^{X-1}Z_tS_{t'}x^{2p+t+t'}\\&&+\sum\limits_{t=0}^{X-1}\sum\limits_{j'=0}^{p-1}Z_tB_{j'}x^{2p+t-j'-1},
\end{IEEEeqnarray*}
then $\deg(h(x))=2p+2X-2$ and $\sum\limits_{j=0}^{p-1}A_{j}B_{j}$ is exactly the coefficient of the monomial $x^{p-1}$ in $h(x)$. Therefore, computing $AB$ is equivalent to retrieving the coefficient of the monomial $x^{p-1}$ in $h(x)$. For $i\in [0, N)$, server $i$ computes $h(a_i)=f(a_i)g(a_i)$.  Then,  applying Proposition \ref{Prop_lag}, we can get the desired result. The proof for being $X$-secure is guaranteed by Corollary \ref{CorXsecure}. An alternative proof can be found in \cite[Section 4.2]{yu2019lagrange}.
\end{proof}

\begin{Remark}
Furthermore, the above scheme can tolerate $N-R_c$ stragglers, provided that the cooperating servers transmit their ``identity'' to the user, who can then broadcast the respective evaluation points to the servers. The evaluation points are small compared to the matrix coefficients, so the additional communication cost is negligible (cf. footnote on p. 4). 
\end{Remark}

\begin{Remark}\label{Remark_mital}
Note that the non-cooperative SDMM scheme in \cite{mital2020secure} uses the  same  kind  of  partitioning  as the proposed cooperative  scheme. However, it cannot mitigate stragglers, since the desired matrix multiplication is a linear combination of \emph{all} the answers, using the fact that the sum of all the $n$-th roots of unity is zero. It can be converted to a cooperative scheme without straggler protection with a download cost  $tr\lceil\frac{p+2X}{X}\rceil$,  which is lower than the one in Theorem \ref{Thm_C1}.  The proof is analogous to that of Theorem \ref{Thm_C1}. Thus, these two alternative cooperative schemes can be seen as choosing between lower communication  cost and straggler protection.
\end{Remark}

{Even though similar techniques can in principle be applied to many other schemes, this will not automatically give a lower download cost. This can be seen from enabling cooperation in the GASP code \cite{d2021degree}, see Examples \ref{gasp-basic}, \ref{gasp} below.}

\begin{Example}\label{gasp-basic}
Here we show an example of the GASP code presented in \cite{d2021degree}. 

Assume that the user is interested in computing $AB$ for $A\in \F_q^{t\times s}$ and $B\in \F_q^{s\times r}$ with the assistance of $N \ge 11$ servers, while any $X = 2$ servers can collude to deduce the information of $A$ and $B$, where $t,r$ are even.
The user divides the matrices $A$ and $B$ into block matrices as
\begin{eqnarray*}
 A=
   \begin{pmatrix} A_{0} \\ A_{1} \end{pmatrix}
   \,,\, B=
     \begin{pmatrix}
       B_{0} & B_{1} \\
   \end{pmatrix}
\end{eqnarray*}
where $A_{j}\in \F_q^{\frac{t}{2} \times s}$ for $j=0, 1$ and $B_{k} \in \mathbb{F}_q^{s \times \frac{r}{2}}$ for $k=0, 1$. Then
\begin{eqnarray*}
  AB &=&  \begin{pmatrix}
    A_{0}B_{0} & A_{0}B_{1} \\
    A_{1}B_{0} & A_{1}B_{1}
  \end{pmatrix}.
\end{eqnarray*}

Let $Z_{0}, Z_{1}$ be two random matrices over $\F_q^{\frac{t}{2} \times s}$ and $S_{0}, S_{1}$ be two random matrices over $\F_q^{s \times \frac{r}{2}}$. The user encodes the matrices $A$ and $B$ by first computing the following polynomials
\begin{eqnarray*}
  f(x) &=& A_0+A_1x+Z_0x^4+Z_1x^6, \\
   g(x) &=& B_0+B_1x^2+S_0x^4+S_1x^5,
\end{eqnarray*}
and then evaluating them at $N$ distinct points in $\F_q$ for some $r \geq 1$, say $a_0, \ldots, a_{N-1}$. The points are chosen such that the system is decodable, which is always possible when $q$ is sufficiently large. Then the user sends $f(a_i)$, $g(a_i)$ to server $i$ for $i \in [0, N)$, who then computes and returns $f(a_i)g(a_i)=h(a_i)$,
\textit{i.e.,} the upload cost is $N(\frac{ts}{2}+\frac{sr}{2})$, where
{\small
\begin{IEEEeqnarray*}{rCl}
   h(x)&=& f(x)g(x)\\
    &=& A_0B_0 + A_1B_0x + A_0B_1x^2 + A_1B_1x^3 \\
    &+& (A_0S_0 + Z_0B_0)x^4 + (A_0S_1 + A_1S_0)x^5 \\
    &+& (A_1S_1 + Z_0B_1 + Z_1B_0)x^6 \\
    &+& (Z_0S_0 + Z_1B_1)x^8 + Z_0S_1x^9 +Z_1S_0x^{10} + Z_1S_1x^{11}.
\end{IEEEeqnarray*}
}Suppose that server $ j_0,  j_1, \ldots,  j_{10}$ be the fastest $11$ servers, then the user gets
\begin{equation}\label{Eqn_GASP_SLE}
\begin{pmatrix}
h(a_{j_0})\\
h(a_{j_1})\\
\vdots\\
h(a_{j_{10}})
\end{pmatrix}\hspace{-1mm}=\hspace{-1mm}\underbrace{\begin{psmallmatrix}
1& a_{ j_0} &\cdots&a_{j_0}^6 &a_{j_0}^8 &\cdots &a_{j_0}^{11} \\
1 & a_{j_1} &\cdots&a_{j_1}^6 &a_{j_1}^8 &\cdots &a_{j_1}^{11} \\
\vdots&\vdots&\ddots&\vdots&\vdots&\ddots&\vdots\\
1 & a_{j_{10}} &\cdots&a_{j_{10}}^6 &a_{j_{10}}^8 &\cdots &a_{j_{10}}^{11} 
\end{psmallmatrix}}_{G'}\otimes I_{\frac{t}{2}\times \frac{t}{2}} \begin{pmatrix}
A_0B_0\\
A_1B_0\\
A_0B_1\\
A_1B_1\\
\vdots\\
\end{pmatrix},   
\end{equation}
from which the first $4$ coefficients (\textit{i.e.}, $AB$) of $h(x)$ can be recovered if $G'$ is nonsingular, then the recovery threshold is $R_c=11$ and the download cost is $\frac{11}{4}tr$.
\end{Example}

\begin{Example}\label{gasp}
Following from the previous example, let $G$ be the inverse of $G'$, and let $G_{i,j}$ denote the $(i,j)$-th entry of $G$. Under the server cooperation model, 
the coefficients $A_0B_0$, $A_1B_0$, $A_0B_1$ and $A_1B_1$ of the terms $1$, $x$, $x^2$, and $x^3$ in $h(x)$ can be retrieved through the following steps.

\begin{table*}[ht]
\begin{center}
\caption{A comparison of the communication costs including both the upload and download costs between the proposed cooperative SDMM codes and previous non-cooperative ones. We assume $p=m^2$ and $m=n$ for convenience.  Mital \emph{et al.} code 1 and its cooperative version cannot mitigate stragglers, {while all the other schemes can}. 
}
\label{comp_upload+download}\setlength{\tabcolsep}{3pt}
\begin{tabular}{|c|c|c|c|}
\hline  & Communication costs&\multirow{2}{*}{$R_c$ (Recovery Threshold)}&\multirow{2}{*}{References} \\
&(Upload cost $+$ Download cost) & & \\
\hline Chang-Tandon code  & $N(\frac{ts}{m}+\frac{sr}{m})+R_c\frac{tr}{m^2}$ & $(m+X)^2$ & \cite{chang2018capacity}\\
\hline Kakar \emph{et al.} code 1  & $N(\frac{ts}{m}+\frac{sr}{n})+R_c\frac{tr}{mn}$ & $(m+X)(n+1)-1$ & \cite{kakar2019capacity}\\
\hline GASP code  & $N(\frac{ts}{m}+\frac{sr}{n})+R_c\frac{tr}{mn}$ & $\ge mn+\max\{m, n\}+2X-1$ & \cite{d2021degree}\\
\hline SGPD code (OPP) & $N(\frac{ts}{m}+\frac{sr}{n})+R_c\frac{tr}{mn}$ & $
 mn+m+nX+2X-1$ & \cite{aliasgari2020private} \\
\hline SGPD code (IPP) & $N(\frac{ts}{p}+\frac{sr}{p})+R_ctr$ & $
 2p+2X-1$ & \cite{aliasgari2020private} \\
\hline Secure Entangled Polynomial code (IPP) & $N(\frac{ts}{p}+\frac{sr}{p})+R_ctr$ & $2p+2X-1$ & \cite{yu2020straggler} \\
\hline  Mital \emph{et al.} code 1&  $N(\frac{ts}{p}+\frac{sr}{p})+R_ctr$ & $p+2X$ ($=N$) & \cite{mital2020secure} \\
\hline The cooperative SDMM code based on Matdot codes in \cite{dutta2020optimal} & $N(\frac{ts}{p}+\frac{sr}{p})+\lceil\frac{R_c}{X}\rceil tr$ & $2p+2X-1$ &Theorem \ref{Thm_C1} \\
\hline The cooperative SDMM code based on Mital \emph{et al.} code 1&  $N(\frac{ts}{p}+\frac{sr}{p})+\lceil\frac{R_c}{X}\rceil tr$ & $p+2X$ ($=N$) & Remark \ref{Remark_mital} \\
\hline
\end{tabular}
\end{center}
\end{table*}

\begin{figure*}[htbp]
\centering
\hspace{-10mm}
\begin{minipage}[t]{0.45\textwidth}
\includegraphics[scale=.6]{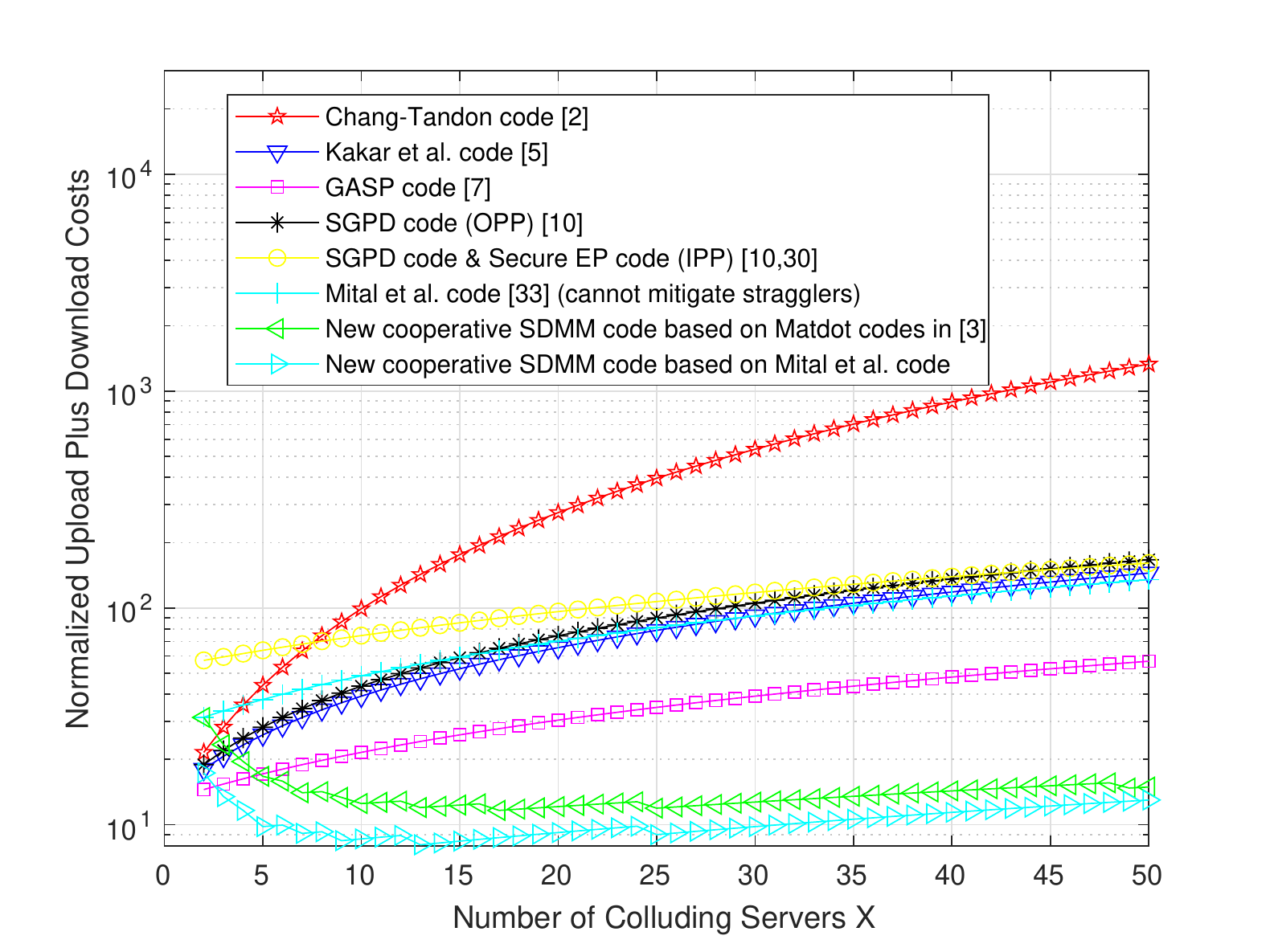}
\end{minipage}
\hspace{5mm}
\begin{minipage}[t]{0.45\textwidth}
\includegraphics[scale=.6]{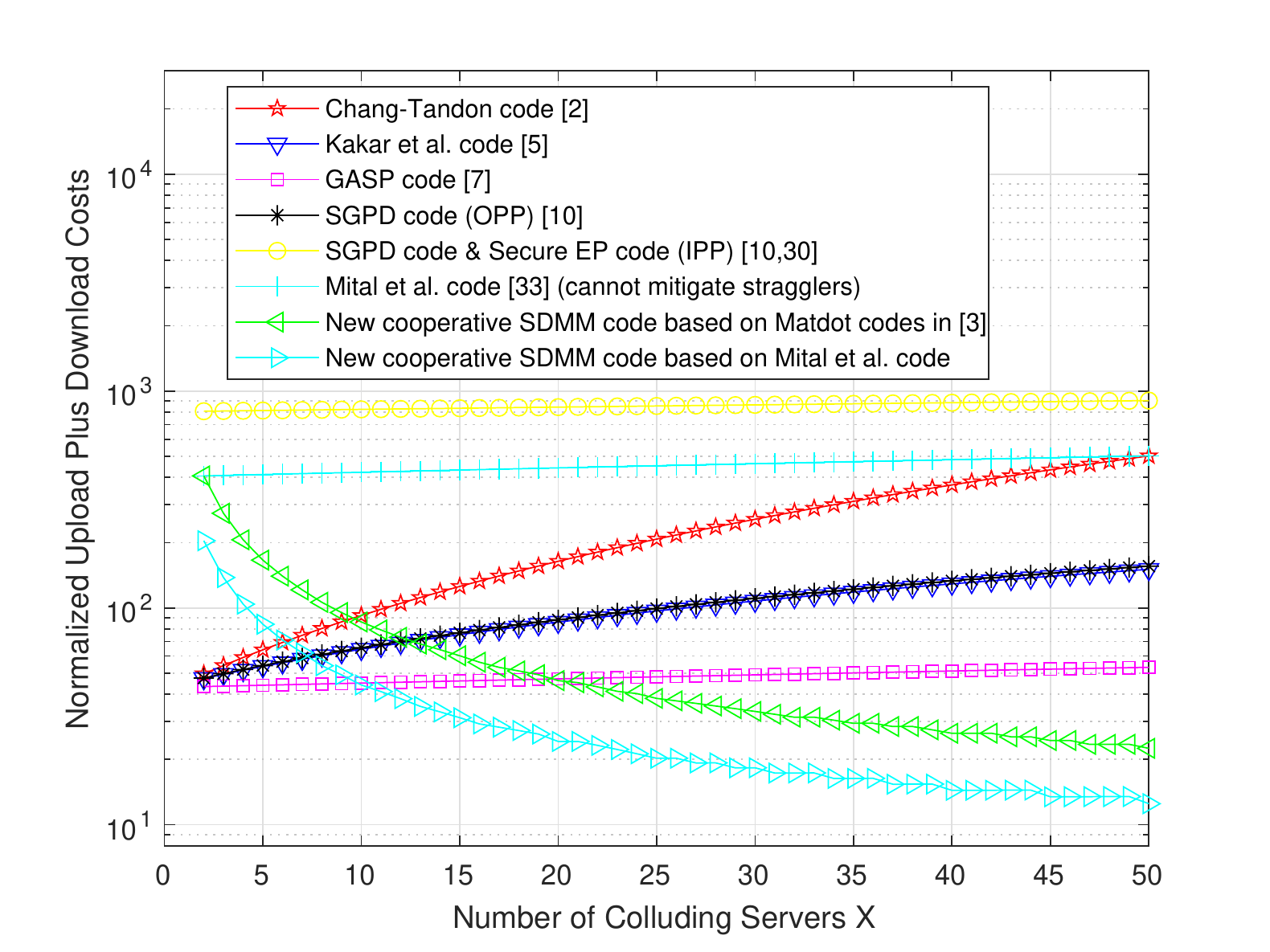}
\end{minipage}
\caption{
Comparison of the normalized communication cost (upload plus download) and the number of colluding servers between the cooperative SDMM code and previous schemes without server cooperation, where $m=5$ (left) and $m=20$ (right),  respectively, and we assume $m=n$, $p=m^2$, and the matrices $A$ and $B$ are square and of the same size (\textit{i.e.}, $t=r=s$) for simplicity. 
The normalization of the communication cost is over $|AB|=|A|=|B|$.}
\label{picture_total}
\end{figure*}

\begin{itemize}
  \item Computation and cooperative phase:  First, server $j$ computes the product of $f(a_j)$ and $g(a_j)$ to obtain $h(a_j)$ for all $j\in [0, N)$. Assume that servers $j_0, j_1, \ldots, j_{10}$ are the fastest $11$ servers, and further assume that servers $j_{2i}$ and $j_{2i+1}$ seek to cooperate with each other for $i\in [0, 5)$. Second, each server $j_i$ ($i\in [0, 11)$)  computes the inverse of $G'$ to obtain $G_{t,i}$, $t\in [0,4)$ and then multiply them with $h(a_{j_i})$.   Finally,  server $j_{2i+1}$ transmits $G_{t,2i+1}h(a_{j_{2i+1}})$, $t\in [0, 4)$,    to server $j_{2i}$ for $i\in [0, \lfloor \frac{11}{2}\rfloor)$, who then further computes
\begin{align*}
 Y_{j_{2i},t}&=G_{t,2i}h(a_{j_{2i}})+G_{t,2i+1}h(a_{j_{2i+1}}), t\in [0, 4).
\end{align*}
\item Decoding phase: Server $j_{2i}$  sends $Y_{j_{2i},t}$,  $t\in [0, 4)$, to the user for $i\in [0, 6)$, where $Y_{j_{10},t}=G_{t,10}h(a_{j_{10}})$.
The user sums up the data they received  and  obtains
\begin{equation*} 
A_{1}B_{0}=Y_{j_0,1}+Y_{j_2,1}+Y_{j_4,1}+Y_{j_6,1}+Y_{j_8,1}+Y_{j_{10},1}
\end{equation*}
according to \eqref{Eqn_GASP_SLE} (after multiplying $G\otimes I_{\frac{t}{2}\times \frac{t}{2}}$ from both sides).  $A_0B_0$,  $A_0B_1$, and $A_1B_1$ can be similarly obtained by calculating
$\sum\limits_{i=0}^{4}Y_{2i,0}$, $\sum\limits_{i=0}^{4}Y_{2i,2}$, and $\sum\limits_{i=0}^{4}Y_{2i,3}$, respectively.
\end{itemize}
Thus, the download cost is $6tr$ and the cooperation cost is $5tr$. Whereas, the download cost is $\frac{11}{4}tr$ in the non-cooperative version. 
\end{Example}
\begin{Remark}
Note that we can also enable the cooperation in other previous SDMM schemes according to Proposition \ref{Prop_lag}, such as the ones mentioned in Table \ref{comp}. However, as the other SDMM schemes use totally different matrix partitioning, then the user needs to retrieve $mn$ coefficients of a specific polynomial, see Example \ref{gasp}. Then,   Proposition \ref{Prop_lag} should be applied $mn$ times to obtain all the $mn$ coefficients. A straightforward calculation easily shows that the download cost of all the schemes in Table \ref{comp} is $tr\lceil\frac{R_c}{X}\rceil$ if cooperation is enabled, but the recovery threshold $R_c$ is different for different schemes. From Table \ref{comp}, we see that the other SDMM schemes can get a gain in the download cost when $X\ge mn$ (or $X\ge m^2$ for Chang--Tandon code) if enabling the cooperation strategy in this work.
Thus, the cooperation strategy in this work may not always be  efficient for the other schemes that use different matrix partitioning, especially if $X<mn$. In this case, we provide a cooperative SDMM scheme based on encryption from a computational secure perspective in the next section.
\end{Remark}

\subsection{Comparison}\label{sec:comp}

In this subsection, we make comparisons of the communication cost (including the upload and download costs) between the proposed cooperative SDMM schemes and some existing ones without cooperation.  Details are provided in Table \ref{comp_upload+download}. 

As different SDMM schemes in the literature employ different matrix partitioning,  to give a fair comparison,  we assume that each server performs the same amount of computations.  
 
 Similarly to \cite{d2021degree}, we consider the OPP given by setting $p=1$ in
\eqref{Eqn_partition_AB}
and the IPP given by
\eqref{Eqn_IPP}.
Clearly, for schemes based on OPP, each server needs $\frac{t}{m}  s\frac{r}{n}=\frac{tsr}{mn}$ scalar multiplications. 
 For schemes based on IPP each server needs $t\frac{s}{p}  r=\frac{tsr}{p}$ scalar multiplications. 
Thus $p=mn$ under the assumption that each server performs the same amount of computations. For convenience, we further assume $m=n$ and thus $p=m^2$.  Figure \ref{picture_total} visualizes the comparison for $m=5$ and $m=20$ by further assuming $N=R_c$. The total communication cost has been normalized with respect to the matrix product size.

From   Figure \ref{picture_total}, we can see that the new SDMM schemes under server cooperation (\textit{i.e.}, the scheme based on the Matdot codes in \cite{dutta2020optimal} and the one based on Mital \textit{et al.} code 1 in \cite{mital2020secure}) can gain the following advantages:

\begin{itemize}
\item The SDMM codes under the server cooperation model always have a significant gain in the communication cost when compared with the ones that employing exactly the same matrix partitioning method (\textit{i.e.,} the SGPD code (IPP) \cite{aliasgari2020private}, the secure Entangled Polynomial code \cite{yu2020entangled}, and Mital \textit{et al.} code 1 \cite{mital2020secure}) as ours. 

\item When compared with the other SDMM schemes that employ totally different matrix partitioning (\textit{i.e.,} OPP), the SDMM codes under server cooperation model have a smaller normalized communication cost when $X$ is larger than a threshold. 
\end{itemize}

Although the cooperative one based on Mital \textit{et al.} code 1 in \cite{mital2020secure} has a smaller communication cost than the one based on the Matdot codes in \cite{dutta2020optimal}, it comes with a penalty that can not mitigate stragglers.) 

\subsection{PIR under  server cooperation model}\label{sec:IT-XTPIR}

In this subsection, we focus on the XTPIR problem as in \cite{jia2020x} but under the server cooperation model.
Suppose there are $m$ files $x^{0}, x^{1}, \ldots, x^{m-1}\in \mathbb{F}_q^{r}$, each of size $r$ stored in the rows of an $m \times r$ matrix $B$. We assume the matrix $B$ is available at the servers. Then the PIR problem of retrieving the $i$--th file is equivalent to securely computing the matrix product $e_i^TB=x^{i}$.
Sometimes it is necessary to divide a file into several stripes, \emph{i.e.},  each file will be of size $x^{i} \in \mathbb{F}_q^{\frac{s}{m} \times r}$ and hence $B \in \mathbb{F}_q^{s \times r}$. In that case the retrieval of a file is achieved by computing the product $(e_i^T \otimes I_{\frac{s}{m}})B=x^i$.
The upload cost is of the order of the number of files, while the download cost is in the order of the size of the files. It is hence usual to ignore upload cost and use download cost as the sole parameter of interest.

Analogously to Theorem \ref{Thm_C1}, we immediately have the following result.

\begin{Theorem}\label{Thm_PIR}
Assume any $X$ servers can collude and  cooperate. Then, there exists an explicit cooperative XTPIR scheme with $X=T$ and download cost  $sr\lceil\frac{2p+2X-1}{X}\rceil$, where each file is of size $sr$, \emph{i.e.,} the PIR rate is $R_c=\frac{1}{\lceil\frac{2p+2X-1}{X}\rceil}$.
\end{Theorem}

\section{Computationally Secure Cooperation}\label{sec:enc_coop}

In this section, we generalize the previous scheme to work in the computationally secure setting and show how SDMM schemes can be constructed with the computationally secure cooperative model.

Particularly, we assume that the collusion and cooperation graphs are independent of each other. The motivation for this comes from the fact that a cooperating server is only exchanging encrypted data, hence the other servers cannot infer any information (it is computationally hard). Colluding servers, for their part, may actively try to infer information and in addition exchange their used seeds, hence violating the protocol. In addition, we assume that all the servers have the capability to cooperate with each other, but only $X$ of them will collude with each other, along the lines discussed above.

\subsection{Cooperative Retrieval with Encryption}
Communication between servers does not  necessarily have to be considered as collusion if the data that is shared is properly encrypted. When a cooperating server is  sharing encrypted data that the other servers cannot decrypt, possible collusion will yield no outcome and is hence equivalent to no collusion.  If servers wish to (successfully) collude, they must break the assumption of being honest-but-curious and actively share data that should remain private.

The next example shows how a one-time pad can be used to secretly recover a shared secret using server cooperation. 

\begin{Example}
Let $y_0, \dots, y_{K-1} \in \F_q^n$ be the recovery set of a linear secret sharing scheme and $\alpha_i \in \F_q$ be the coefficients for recovery. Then the secret can be reconstructed as $m = \sum_i \alpha_i y_i$, where the $\alpha_i$'s are assumed to be known to all parties. Using cooperating servers, the secret can be recovered by downloading from just 1 server by offloading some of the communication to the servers.

Each server encrypts their share $y_i$ using a one-time pad. Hence, they get the value $z_i = y_i + r_i$, where $r_i$ is independently chosen uniformly at random and only known to server $i$. The value $z_i$ can be shared freely to anyone not knowing the key $r_i$ without leaking information. Each server sends their $z_i$ to a specified server that computes the linear combination
\begin{equation*}
    \sum_i \alpha_i z_i
\end{equation*}
and sends this to the user. Now the user can compute
\begin{equation*}
    \sum_i \alpha_i z_i - \sum_i \alpha_i r_i = \sum_i \alpha_i y_i = m.
\end{equation*}
To do this the user needs to know the one-time pads $r_i$. 
\end{Example}

To know the one-time pads of each server, the user needs to download them from each server securely. If such a secure download were possible, then it would be more convenient to just download $y_i$ instead, since $r_i$ is no smaller than $y_i$. Therefore, it is not practical to use a one-time pad to efficiently recover a shared secret. The following Proposition combines Proposition \ref{Prop_lag} with the encryption scheme described in Construction \ref{Con_enc}.

\begin{Proposition}\label{Prop_enc}
Assume there are $N$ servers that can all cooperate. Assume that $y_i$ is the evaluation of some polynomial $L(x)$ at $x_i$ for $i \in [0, N)$, where $\deg(L(x)) = K-1\in[0, N)$. If server $i$ has the data $x_0, \dots, x_{N-1}$ and $y_i$ for $i \in [0, N)$, then the user can obtain one of the coefficients $L_\theta$, where $\theta \in [0, K)$, with the download cost being $|L_\theta|$.
\end{Proposition}

\begin{proof}
The coefficient can be obtained using encryption and cooperation like follows.
\begin{itemize}
    \item Encryption phase: Each server $i \in [0, N)$ chooses a random key $k_i$  and uses a pseudorandom function to compute a random value $r_i$ of the same size as $y_i$. 
The pseudorandomness property of the PRF means that the entries in $r_i$ are (pseudo) uniformly distributed. Server $i$ then computes $z_i = y_i + r_i$ and transmits this to a representative server.
    Server $i$ also transmits their secret key $k_i$ to the user. 
    \item Cooperation phase: W.l.o.g., assume that the $K$ first servers are the fastest ones and the coefficient is obtained from their answers. Once the representative server has received the $K-1$ first responses and their own, they compute the polynomial $\ell^{(i)}(x)$ in \eqref{Eqn_ell} for $i \in [0, K)$ using the points $x_0, \dots, x_{N-1}$, and then compute the value
    \begin{equation*}
        \sum_{i=0}^{K-1} \ell^{(i)}_\theta z_i.
    \end{equation*}
    The representative server then transmits this to the user.
    \item Decryption phase: Using the secret keys $k_i$ for $i \in [0, K)$ the user is able to compute each $r_i$ using the same PRF the servers used. Then the user computes the polynomials $\ell^{(i)}(x)$ in \eqref{Eqn_ell} for $i \in [0, K)$ and computes
\begin{equation*}
    \sum_{i=0}^{K-1} \ell^{(i)}_\theta z_i - \sum_{i=0}^{K-1} \ell^{(i)}_\theta r_i = \sum_{i=0}^{K-1} \ell^{(i)}_\theta y_i = L_\theta
\end{equation*}
according to \eqref{Eqn_L_theta}.
\end{itemize}

The total download cost for the user is $|L_\theta|$ and the sizes of the keys, which are assumed to be small compared to $|L_\theta|$ and the cooperation cost is $(K-1)|L_\theta|$, since the representative server downloads from $K-1$ other servers.
\end{proof}

The total download cost for the users and the servers is now $K|L_\theta|$, which is the same as without cooperation. Hence, no additional communication is introduced by using Proposition \ref{Prop_enc}. The computational complexity is increased slightly from the non-cooperative scheme, since each server has to generate the pseudorandom matrix and add that to the share. Additionally, the collector server needs to compute the linear combination. Finally, the user needs to compute the linear combination on the pseudorandom matrices and subtract that from the result. This increase is only minor if there is an efficient way of producing pseudorandom matrices.

\subsection{Example of the GASP Code under Server Cooperation with Encryption}

Following Example \ref{gasp-basic}, under the server cooperation model from a computational secure perspective,
the coefficients $A_0B_0$, $A_1B_0$, $A_0B_1$ and $A_1B_1$ of the terms $1$, $x$, $x^2$, and $x^3$ in $h(x)$ can be retrieved through the following steps.
\begin{itemize}
    \item Encryption phase: First, server $j$ computes the product of $f(a_j)$ and $g(a_j)$ to obtain $h(a_j)$ for all $j \in [0, N)$. Assume that servers $j_0, j_1, \ldots, j_{10}$ are the fastest $11$ servers. Each server $j_i$ then chooses a uniformly random key $k_{j_i}$ and uses that to compute a pseudorandom value $r_{j_i} \in \F_q^{\frac{t}{2} \times \frac{r}{2}}$ using a pseudorandom function. 
    \item Cooperation phase: Server $j_i$ transmits $h(a_{j_i}) + r_{j_i}$ to server $j_0$ and $k_{j_i}$ to the user. The cooperation cost is $10\cdot\frac{tr}{4} = \frac{5}{2}tr$. Server $j_0$  obtains
\begin{equation} 
\begin{psmallmatrix}
h(a_{j_0})+r_{j_0}\\
h(a_{j_1})+r_{j_1}\\
\vdots\\
h(a_{j_{10}})+r_{j_{10}}
\end{psmallmatrix}={G'}\otimes I_{\frac{t}{2}\times \frac{t}{2}} \begin{psmallmatrix}
A_0B_0\\
A_1B_0\\
A_0B_1\\
A_1B_1\\
\vdots\\
\end{psmallmatrix}+\begin{psmallmatrix}
r_{j_0}\\
r_{j_1}\\
\vdots\\
r_{j_{10}}
\end{psmallmatrix},   
\end{equation}
where the left hand side is known. By calculation the inverse $G$ of $G'$, server $j_0$  obtains
\begin{equation}\label{Eqn_Inverse_CS} 
\begin{split}
&(G\otimes I_{\frac{t}{2}\times \frac{t}{2}})\begin{pmatrix}
h(a_{j_0})+r_{j_0}\\
h(a_{j_1})+r_{j_1}\\
\vdots\\
h(a_{j_{10}})+r_{j_{10}}
\end{pmatrix} \\
&=\begin{pmatrix}
A_0B_0\\
A_1B_0\\
A_0B_1\\
A_1B_1\\
\vdots\\
\end{pmatrix}+(G\otimes I_{\frac{t}{2}\times \frac{t}{2}})\begin{pmatrix}
r_{j_0}\\
r_{j_1}\\
\vdots\\
r_{j_{10}}
\end{pmatrix},  
\end{split}
\end{equation}
Server $j_0$ sends the first four block entries of the vector in the LHS of \eqref{Eqn_Inverse_CS}  to the user. Each block has size $\frac{tr}{4}$, so the total size is $tr$.
    \item Decryption phase: The user uses the secret keys $k_{j_i}$ to compute the pseudorandom matrices $r_{j_i}$ and computes the first four block entries of 
    \begin{equation}\label{Eqn_gasp_random}
        (G\otimes I_{\frac{t}{2}\times \frac{t}{2}})\begin{pmatrix}
r_{j_0}\\
r_{j_1}\\
\vdots\\
r_{j_{10}}
\end{pmatrix}
    \end{equation}
    Then the user is able to obtain $AB$ by subtracting equations \eqref{Eqn_Inverse_CS} and \eqref{Eqn_gasp_random} and rearranging the blocks.

\end{itemize}
Thus, the download cost for the user is $tr$ plus the size of the secret keys, which can be considered negligible (some hundreds of bits) compared to the (presumably large) matrices.

\subsection{Encryption-based Cooperative SDMM Scheme}

The encryption-based cooperation can be adopted to other SDMM schemes as well by applying Proposition \ref{Prop_enc}, similarly to the previous example. 

\begin{Theorem}\label{Thm_enc_SDMM}
Consider an non-cooperative SDMM scheme with recovery threshold $R_c$, upload cost $C_u$, and download cost $C_d=trR_c$. If all servers can cooperate then there is a cooperative SDMM scheme with upload cost $C_u$, download cost $tr$, cooperation cost $\frac{R_c - 1}{R_c}C_d$, and recovery threshold $R_c$, while the scheme is computationally secure.
\end{Theorem}

\begin{proof}
The beginning of the new scheme is the same as the original scheme. The server $j$ receives the encoded matrices $\Tilde{A}_j$ and $\Tilde{B}_j$, which are multiplied to get $\Tilde{A}_j\Tilde{B}_j$. Let $j_0, \dots, j_{R_c - 1}$ be the fastest $R_c$ servers that are used for the recovery of the answer. Server $j_i$ encrypts $\Tilde{A}_i\Tilde{B}_i$ using Construction \ref{Con_enc}. The ciphertext has the same size as the product $\Tilde{A}_{j_i}\Tilde{B}_{j_i}$. The cost of transmitting the encryptions to server $j_0$ is $\frac{R_c - 1}{R_c}C_d$, since server $j_0$ doesn't need to send anything to itself. Server $j_0$ then performs the interpolation with the ciphertexts and gets a result of size $tr$. This is then transmitted to the user who can recover the product $AB$ by interpolating with the pseudorandom matrices from the encryption and subtracting the results.

In addition to the communication described above, the keys used in the encryption and the evaluation points need to be communicated. However, the sizes are small and not proportional to the sizes of the matrices, so we ignore them here.

The new cooperative scheme is $X$-secure until the cooperative step, since the original scheme is also $X$-secure. The encryption scheme is computationally secure according to Definition \ref{Def_ind-cpa}, so the $X$-security is not broken if the adversary is assumed to be computationally bounded.
\end{proof}

\begin{Remark}
If all servers cannot cooperate with each other, then the above scheme can be modified so  that one representative server from each cooperating set clusters works together with the other representatives, similar to Theorem \ref{Thm_C1}. Then the download cost is $Ctr$, where $C$ is the number of cooperating clusters, and the cooperation cost is $\frac{R_c - C}{R_c} C_d$. Notice that the cooperating sets and colluding sets need not be the same, as is assumed in Theorem \ref{Thm_C1}.
\end{Remark}

\begin{Corollary}\label{Cor_enc_matdot}
The secure MatDot code used in Theorem \ref{Thm_C1} can be converted to a cooperative SDMM scheme with upload cost $N(\frac{ts}{p} + \frac{sr}{p})$, download cost $tr$, cooperation cost $(2p+2X-2)tr$, and recovery threshold $R_c = 2p + 2X - 1$.
\end{Corollary}

\section{Conclusions and Future Work}\label{sec:conclusion}

In this paper, we considered  a new cooperative SDMM model, which utilizes a possible link between the helper servers to reduce the download cost for the user. More precisely, this is enabled by outsourcing server-to-user communication to the servers. 
Outsourcing communication in addition to computation may also help in preventing network congestion. In some cases, \emph{e.g.}, when servers are clustered within close proximity,  inter-server communication can be considered to be cheaper than server--user communication.  Based on this model, a new information--theoretically secure distributed matrix multiplication scheme  was proposed. A comparison of the key parameters between the proposed cooperative SDMM code and some previous ones was given, showing  a significant gain in the download cost. 
While server cooperation can be seen as a general strategy, explicit per-scheme description is nontrivial as there are several parameters to consider.  Construction of information--theoretically secure schemes that allow for a more general matrix partitioning is part of our ongoing work.

By further assuming collusion and cooperation graphs are independent of each other, more specifically, assuming all the servers can cooperate, we proposed a computationally secure SDMM scheme based on encryption, which allows for a more general matrix partitioning and  achieves  yet better communication cost.

\section*{Acknowledgment}
The authors would like to thank the Guest Editors and the three anonymous reviewers for their valuable suggestions and comments, which have greatly improved the presentation and quality of this paper.	

\bibliographystyle{IEEEtran}
\bibliography{SDMM}

%
%
\begin{IEEEbiography}[{\includegraphics[width=1in,height=1.25in,clip,keepaspectratio]{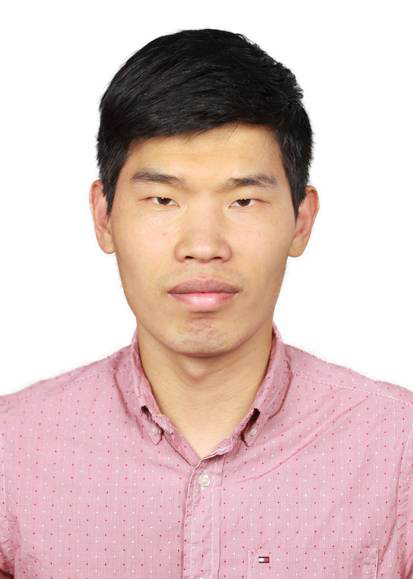}}]
{Jie Li}(Member, IEEE)    
received the B.S. and M.S. degrees in mathematics from Hubei University, Wuhan, China, in 2009 and 2012, respectively, and received the Ph.D. degree from the department of communication engineering, Southwest Jiaotong University, Chengdu, China, in 2017. From  2015 to 2016, he was a visiting Ph.D. student with the Department of Electrical Engineering and Computer Science, The University of Tennessee at Knoxville, TN, USA.  From   2017 to   2019, he was a postdoctoral researcher with the Department of Mathematics, Hubei University, Wuhan, China. From  2019 to 2021, he was a postdoctoral researcher with the Department of Mathematics and Systems Analysis, Aalto University, Finland. He is currently a senior researcher with the Theory Lab, Huawei Tech. Investment Co., Limited, Hong Kong SAR, China. His research interests include private information retrieval, coding for distributed storage, and sequence design.

Dr. Li received the IEEE Jack Keil Wolf ISIT Student Paper Award in 2017.
\end{IEEEbiography}

\begin{IEEEbiography}[{\includegraphics[width=1in,height=1.25in,clip,keepaspectratio]{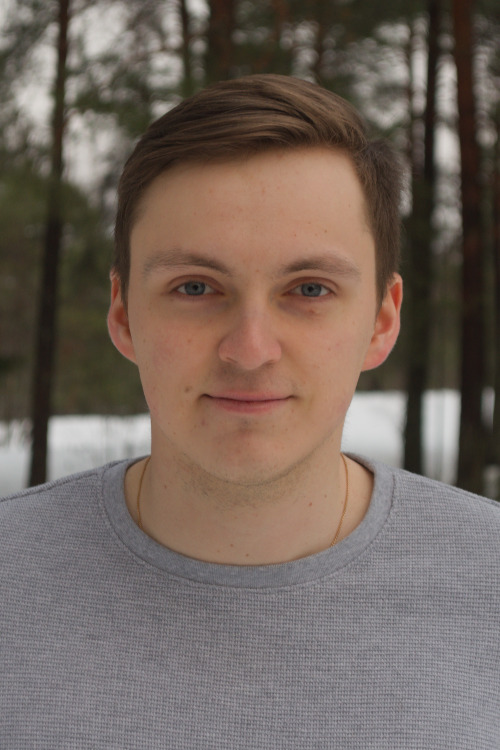}}]
{Okko Makkonen}
received the B.Sc. (Tech.) degree in mathematics from Aalto University, Finland, in 2021, where he is currently pursuing the M.Sc. (Tech.) degree in Hollanti's ANTA group. His research interests include information-theoretically secure distributed computation schemes.
\end{IEEEbiography}

\begin{IEEEbiography}[{\includegraphics[width=1in,height=1.25in,clip,keepaspectratio]{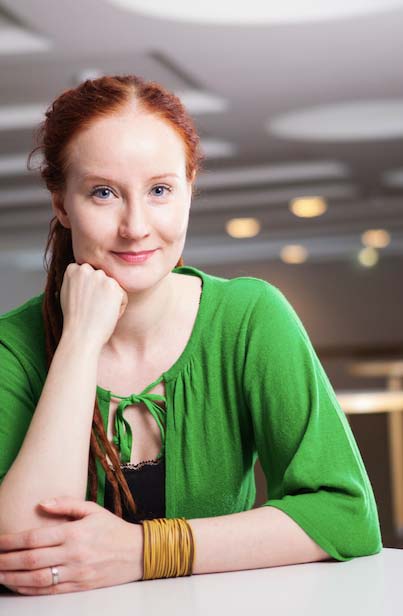}}]
{Camilla Hollanti}(Member, IEEE)     received the M.Sc. and Ph.D. degrees from the University of Turku, Finland, in 2003 and 2009, respectively, both in pure mathematics. Her research interests lie within applications of algebraic number theory to wireless communications and physical layer security, as well as in combinatorial and coding theoretic methods related to distributed storage systems and private information retrieval.

For 2004-2011 Hollanti was with the University of Turku. She joined the University of Tampere as  Lecturer for the academic year 2009-2010. Since 2011, she has been with the Department of Mathematics and Systems Analysis at Aalto University, Finland, where she currently works as Full Professor and Vice Head, and leads a research group in Algebra, Number Theory, and Applications. During 2017-2020, Hollanti was affiliated with the Institute of Advanced Studies at the Technical University of Munich, where she held a three-year Hans Fischer Fellowship, funded by the German Excellence Initiative and the EU 7th Framework Programme.

Hollanti is currently an editor of the AIMS Journal on Advances in Mathematics of Communications,  SIAM Journal on Applied Algebra and Geometry, and IEEE Transactions on Information Theory. She is a recipient of several grants, including six Academy of Finland grants. In 2014, she received the World Cultural Council Special Recognition Award for young researchers. In 2017, the Finnish Academy of Science and Letters awarded her the V\"ais\"al\"a Prize in Mathematics. For 2020-2022, Hollanti is serving as a member of the Board of Governors of the IEEE Information Theory Society, and is one of the General Chairs of IEEE ISIT 2022.
\end{IEEEbiography}

\begin{IEEEbiography}[{\includegraphics[width=1in,height=1.25in,clip,keepaspectratio]{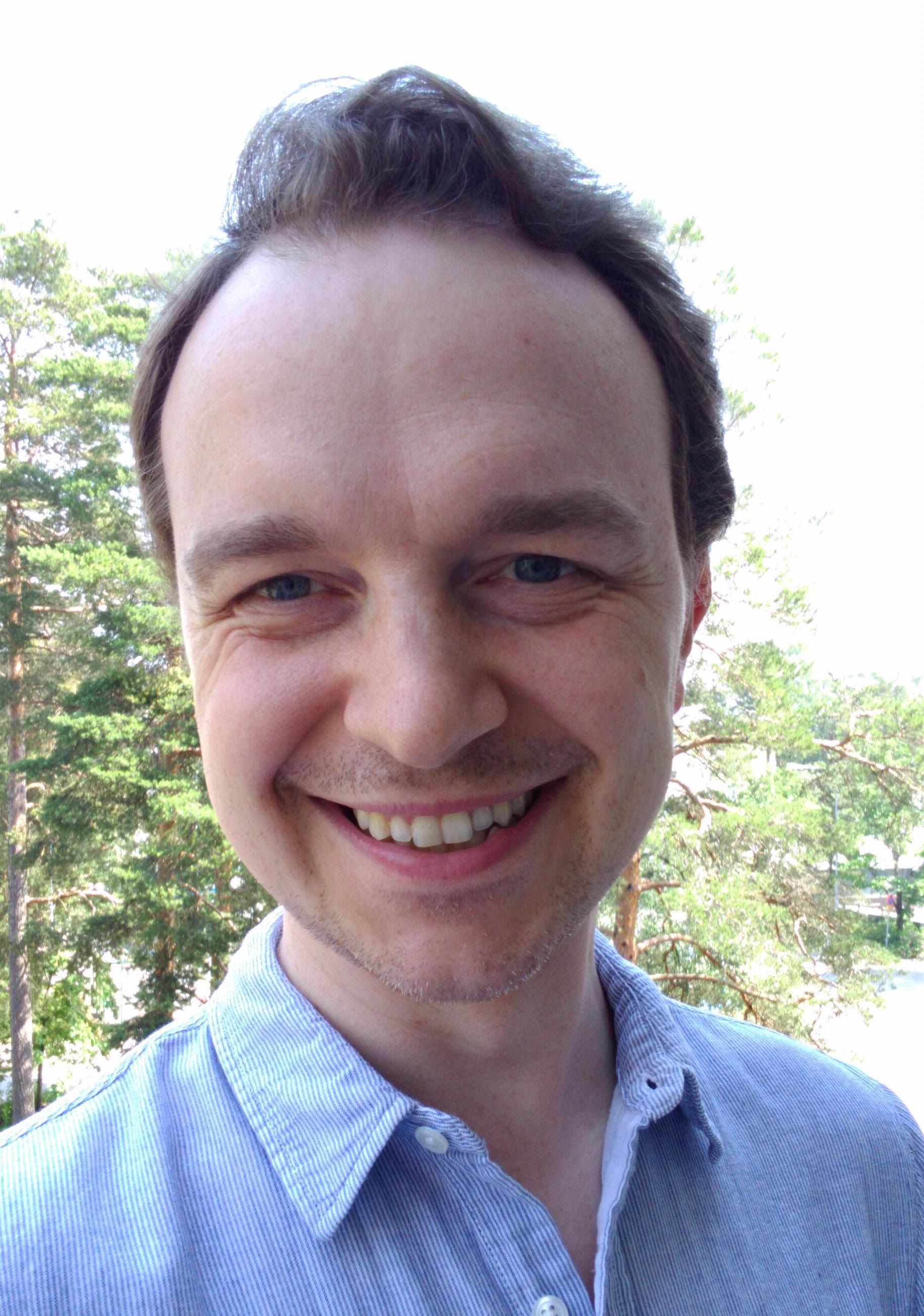}}]
{Oliver W. Gnilke} received the Diploma degree in mathematics from the Universit{\"a}t Hamburg, Germany, in 2010, and the Ph.D. degree from the University College Dublin, Ireland, in 2015. He went on to work as a Post-Doctoral Researcher with Aalto University from 2015 to 2018. He received a full year grant by the Finnish Cultural Foundation and was a Visiting Researcher with the TU Munich for two months. Since 2019, he has been working with Aalborg University, Denmark, where he is currently an Associate Professor for coding theory and cryptography with the Department of Mathematical Sciences. His research interests include applications of coding theory, privacy and security in communications, and combinatorial
designs.  
\end{IEEEbiography}
\end{document}